\documentclass[conference,letterpaper]{IEEEtran}

\usepackage[utf8]{inputenc} 
\usepackage[T1]{fontenc}
\usepackage{url}
\usepackage{ifthen}
\usepackage{cite}
\usepackage[cmex10]{amsmath} 

\usepackage{amssymb}
\usepackage{amsfonts}
\usepackage{graphicx}
\usepackage{subcaption}
\usepackage{microtype}
\usepackage{flushend}
\usepackage[colorlinks,linkcolor=blue]{hyperref}

\newtheorem{remark}{Remark}
\newtheorem{theorem}{Theorem}
\newtheorem{lemma}{Lemma}

\newenvironment{proof}{{\noindent\it Proof:}}{\hfill $\blacksquare$\par}
\newcommand{\dd}{\mathrm{d}}

\begin{document}
\title{Exact SINR Analysis of Matched-filter Precoder} 

 \author{%
  \IEEEauthorblockN{Hui Zhao, Dirk Slock, and Petros Elia}
  \IEEEauthorblockA{Communication Systems Department,
                    EURECOM,
                    Sophia Antipolis, France\\
                    Email: hui.zhao@eurecom.fr; dirk.slock@eurecom.fr; petros.elia@eurecom.fr}
 }



\maketitle


\begin{abstract}
   This paper answers a fundamental question about the exact distribution of the signal-to-interference-plus-noise ratio (SINR) under matched-filter (MF) precoding. Specifically, we derive the exact expressions for the cumulative distribution function (CDF) and the probability density function (PDF) of SINR under MF precoding over Rayleigh fading channels. Based on the exact analysis, we then \emph{rigorously} prove that the SINR converges to some specific distributions separately in high SNR and in massive MIMO. To simplify the exact result in general cases, we develop a good approximation by modelling the interference as a Beta distribution. We then shift to the exact analysis of the transmit rate, and answer the fundamental question: How does the exact rate converge to the well-known asymptotic rate in massive MIMO? After that, we propose a novel approximation for the ergodic rate, which performs better than various existing approximations. Finally, we present some numerical results to demonstrate the accuracy of the derived analytical models.
\end{abstract}

\section{Introduction}
In multiuser multiple-input multiple-output (MU-MIMO) systems, a base station (BS) equipped with multiple antennas simultaneously serves several users, each requesting different/same messages, in the same time-frequency resource. Some processing techniques on the original data signal need to be implemented at the BS (i.e., precoding)  to manage the inter-user interference at each receiver in such spatial multiplexing systems. While very considerable research has focused on a variety of advanced precoding schemes, the workhorses are linear precoders such as matched filtering (MF) \cite{Larsson_Mag,Wagner,Rusek}. These linear precoders maintain low complexity while achieving spectral efficiencies often close to the non-linear Dirty-Paper Coding \cite{Caire_Dirty}, especially in large-scale antenna arrays \cite{Hien}. However, there are some fundamental questions about linear precoding, especially about MF, that remain unanswered.

It is the case though that the exact distribution of the signal-to-interference-plus-noise ratio (SINR) under MF precoding has been not revealed due to $i)$ the correlation between the interference and the useful signal and $ii)$ the correlation among the interference terms. This also makes the exact analysis of the ergodic rate intractable in the fast-fading scenario.  For that, various approximations of the ergodic rate under MF precoding have been developed in some special cases, especially in the massive MIMO regime. For example, the authors of \cite{Hien,Zhao_SPAWC} applied Jensen's Inequality to derive a lower bound of the ergodic rate, while the so-called ``near deterministic" method is proposed in \cite{Yeon_Geun_Lim,Mi_TCOM} to approximate the ergodic rate, especially in high and low signal-to-noise ratio (SNR) regimes. The authors of \cite{Zhang,Zhao_VectorCC,Zhao_LMS} developed a good approximation method for the ergodic rate though with the requirement of large numbers of both transmit antennas and served users. Again in the massive MIMO regime, some works (e.g., \cite{Hoydis2013,Sadeghi,Emil_unbounded}) treated the interference as noise to alternatively analyze an achievable rate so that the asymptotic deterministic equivalence of this achievable rate can be derived. Although the aforementioned various approximations have provided us with good analytical models for performance evaluation of MF, none of them tells us about the answer to the fundamental question: To what extent does the exact rate converge to the approximations, e.g., the convergence type? It is also the case that few works analyze the SINR randomness to investigate the outage performance of MF precoding in slow-fading \cite{Feng}. We note that the authors of \cite{Feng} also assumed a large-scale antenna array to ignore the correlation between the useful signal and the interference to simplify the SINR analysis. 

In this paper, we will present an exact analysis of the SINR distribution under MF precoding \emph{for the first time}, as well as various simplified results in some special cases, including high SNR and massive MIMO. More importantly, we will \emph{rigorously} prove that the exact SINR converges to these simplified distributions associated with the corresponding special cases. We will also develop a good approximation to simplify the SINR analysis in general cases. Based on the exact SINR analysis,  we will \emph{rigorously} prove that the exact transmit rate converges to the well-known asymptotic rate in massive MIMO (cf. \cite{Rusek})  \emph{almost surely}. Finally, a novel and simple approximation for the ergodic rate will be presented, which is shown to be more accurate than various existing approximations.\footnote{\emph{Notations:} $||\cdot||$ denotes the norm-2 of a vector, while $|\cdot|$ denotes the magnitude of a complex number. For a matrix $\bf A$, we use ${\bf A}^H$, ${\bf A}^*$ and ${\bf A}^T$ to denote its conjugate transpose, conjugate part and non-conjugate transpose respectively. $\text{Exp}(\cdot)$, $\text{Gamma}(\cdot,\cdot)$, $\text{Inv-Gamma}(\cdot,\cdot)$, $\text{Beta}(\cdot,\cdot)$ and $\mathcal{CN}(\cdot,\cdot)$ denote the exponential distribution, the Gamma distribution, the inverse Gamma distribution, the Beta distribution and the complex Gaussian distribution respectively. $\stackrel{d.}{\longrightarrow}$, $\stackrel{p.}{\longrightarrow}$ and $\stackrel{a.s.}{\longrightarrow}$ denote the convergence in distribution, the convergence in probability and the almost sure convergence respectively.}

\section{System Model and Exact SINR Analysis}
In a downlink MU-MIMO system, a BS equipped with $L$ antennas adopts MF precoding to serve $K$ single-antenna users at a time. For the channel matrix ${\bf H}  \in \mathbb{C}^{L \times K}$, the MF precoder is of the form
\begin{align}
    {\bf W} = \bigg[ \frac{{\bf h}_1^*}{|| {\bf h}_1^* ||},  \frac{{\bf h}_2^*}{|| {\bf h}_2^* ||}, \cdots, \frac{{\bf h}_K^*}{|| {\bf h}_K^* ||}   \bigg] \in \mathbb{C}^{L \times K},
\end{align}
where ${\bf h}_k \in  \mathbb{C}^{L \times 1}$, the $k$-th row of ${\bf H}$, denotes the channel vector from the BS to the $k$-th user. Then, under the transmit power constraint $P_t$, the signal transmitted at the BS is ${\bf s} = \sqrt{P_t/K}{\bf W} {\bf x}$, where the $k$-th element $x_k$ of ${\bf x} \in \mathbb{C}^{K \times 1}$ is the data signal intended by the $k$-th user. Under the usual Gaussian signalling assumption, the SINR at the $k$-th user for decoding $x_k$ takes the form
\begin{align}\label{SINR_def}
    \text{SINR}_k = \frac{ \frac{P_t}{K} ||{\bf h}_k||^2 }{ \sigma_k^2 + \frac{P_t}{K} \sum_{{i=1, i \neq k}}^K  \frac{ | {\bf h}_k^T {\bf h}_i^* |^2 }{|| {\bf h}_i^* ||^2} },
\end{align}
where $\sigma_k^2$ denotes the power of the additive Gaussian white noise (AWGN) at the $k$-th user. The transmission rate for the $k$-th user then takes the form
\begin{align}
    R_k = \ln \big( 1 + \text{SINR}_k \big) \ \text{nats/s/Hz}.
\end{align}
The ergodic rate $\bar R_k$ is defined as the mean of $R_k$ averaged over channel states.

In the remainder of the paper, we consider Rayleigh fading channels for performance analysis \cite{three_GPP}, where the elements in ${\bf H}$ are independently complex Gaussian distributed with zero-mean and unit-variance.
To facilitate the analysis, we rewrite $\text{SINR}_k$ in \eqref{SINR_def} as
\begin{align}\label{SINR_divide_eq}
    \text{SINR}_k 
    = \Big(  Y +  \sum\nolimits_{i=1, i \neq k}^K  X_i \Big)^{-1},
\end{align}
where $Y \triangleq \frac{K \sigma_k^2}{P_t||{\bf h}_k||^2}$,  $X_i \triangleq \frac{ 1}{|| {\bf h}_i^* ||^2}   \big| {\bf u}_k^T {\bf h}_i^*  \big|^2$, and ${\bf u}_k \triangleq \frac{ {\bf h}_k  }{||{\bf h}_k||}$ with unit-norm. It is easy to prove that $||{\bf h}_k||^2$ follows a Gamma distribution with the shape parameter $L$ and the scale parameter $1$, denoted by $||{\bf h}_k||^2 \sim \text{Gamma}(L,1)$. So we have that $Y \sim \text{Inv-Gamma}(L,\frac{K \sigma_k^2}{P_t})$.
According to Cauchy–Schwarz Inequality, we know that $X_i  \le 1$.
In the following, we analyze the distribution of $X_i$ in Lemma~\ref{Xi_Distri_Lem}, where the cumulative distribution function (CDF) and the probability density function (PDF) are derived.

\begin{lemma}\label{Xi_Distri_Lem}
    The random variable $X_i\triangleq \frac{ 1}{|| {\bf h}_i^* ||^2}   \big| {\bf u}_k^T {\bf h}_i^*  \big|^2 \in [0,1]$ has a Beta distribution with the first shape parameter $1$ and the second shape parameter $L-1$, denoted by $X_i \sim \text{Beta}(1,L-1)$. The CDF and PDF of $X_i$ are respectively
    \begin{align}\label{CDF_Xi_eq}
    &F_{X_i}(x)  = 1 - (1-x)^{L-1},  \\
    &f_{X_i}(x) = (L-1) (1-x)^{L-2}.
    \end{align}
Furthermore, $X_i$ is independent of ${\bf u}_k$ (or equivalently, ${\bf h}_k$), and $\{X_i\}_{i=1,i\neq k}^K$ are independent of each other.
\end{lemma}

\begin{proof}
    The proof is relegated to Appendix~\ref{Proof_Xi_Distri_Lem}.
\end{proof}

Let $\text{Im}\{\cdot\}$ denote the imaginary part of a complex number, and let $\jmath \triangleq \sqrt{-1}$ represent the imaginary unit. We use $\mathcal{K}_\cdot(\cdot)$, $\Gamma(\cdot)$ and $\Upsilon(\cdot,\cdot)$ to respectively denote the modified Bessel function of the 2nd kind, the Gamma function and the lower incomplete Gamma function \cite{Gradshteyn}. Now, we present the distribution of $\text{SINR}_k$ in Theorem~\ref{SINR_MF_dist_thm}, which is derived based on the characteristic function (CF) of $Y + \sum_{i=1,i\neq k}^K X_i$.
\begin{theorem}\label{SINR_MF_dist_thm}
    The CDF and PDF of $\text{SINR}_k$ under MF precoding are given by \eqref{CDF_SINRk_eq} and \eqref{PDF_SINRk_eq} respectively, shown at the top of the next page. 
    \begin{figure*}
    \begin{align}
        &F_{\text{SINR}_k}(\gamma) = \frac{1}{2} + \frac{2(L-1)^{K-1}}{\pi \Gamma(L)} \Big(\frac{K\sigma_k^2 }{P_t} \Big)^{\!\! \frac{L}{2}} \! \int_0^\infty \text{Im}\bigg\{   \frac{ \jmath^L \exp\big(\jmath t (K-1 - \frac{1}{\gamma})  \big)} {t (\jmath t)^{(L-1)(K-1) - L/2} }  \mathcal{K}_L\!\bigg(  \sqrt{-4\jmath  K\sigma_k^2 t/P_t}  \bigg)   \big[ \Upsilon(L-1, \jmath t) \big]^{K-1} \! \bigg\} \dd t \label{CDF_SINRk_eq}\\
         &f_{\text{SINR}_k} (\gamma) = \frac{(L-1)^{K-1}}{ \gamma^2 \pi \Gamma(L)} \Big(\frac{K\sigma_k^2 }{P_t} \Big)^{\!\! \frac{L}{2}}\int_{-\infty}^{+\infty}   \frac{ \jmath^L \exp\big(\jmath t (K-1 - \frac{1}{\gamma})  \big)} {(\jmath t)^{(L-1)(K-1)-L/2} } \mathcal{K}_L\bigg( \sqrt{-4\jmath  K\sigma_k^2 t/P_t} \bigg)   \big[ \Upsilon(L-1, \jmath t) \big]^{K-1} \dd t \label{PDF_SINRk_eq}
    \end{align}
    \rule{18cm}{0.01cm}
    \end{figure*}
\end{theorem}

\begin{proof}
    The proof is relegated to Appendix~\ref{Proof_SINR_MF_dist_thm}.
\end{proof}

\section{Simplified Results of SINR Distribution}
In this section, we will perform some analysis to simplify the results of the SINR distribution.

\subsection{High SNR Regime}
We first consider the high SNR regime where $P_t$ goes to infinity.
Let $X' \triangleq \frac{1}{K-1} \sum_{i=1, i \neq k}^K X_i \in [0,1]$. Now, we have the following result.

\begin{lemma}\label{SINRk_dist_limit_lem}
   In high SNR,  we have the convergence result: 
   \begin{align}\label{SINRk_dist_limit}
       \frac{1}{(K-1) \text{SINR}_k} \stackrel{d.}{\longrightarrow}
        X', \text{   as } P_t \to \infty
   \end{align}
\end{lemma}

\begin{proof}
We first rewrite the CDF of $\text{SINR}_k$ in \eqref{CDF_general} as \eqref{CDF_approx}, shown at the top of the next page, where $(a)$ follows from the CDF of $Y$ which follows an inverse Gamma distribution, and where  $f_{X'}(x)$ denotes the PDF of $X'$. In \eqref{CDF_approx},  $\xi \triangleq \min\big\{ \frac{1}{(K-1)\gamma}, 1 \big\}$ for any $\gamma>0$.
\begin{figure*}
\begin{align}\label{CDF_approx}
    F_{\text{SINR}_k} (\gamma) \!
    =\!  \Pr\left\{ Y \ge \Big( \frac{1}{\gamma} - (K-1) X' \Big) \right\} 
    \overset{(a)}{=} 1 - \frac{1}{\Gamma(L)} \int_0^{\xi} \Gamma\bigg(L, \frac{K \sigma_k^2}{P_t} \Big( \frac{1}{\gamma} - (K-1) x \Big)^{-1} \bigg)  f_{X'} (x) \dd x
\end{align}
\rule{18cm}{0.01cm}
\end{figure*}
As $x \to 0$, we have  that $\lim_{x\to 0} \frac{\Gamma(L,x)}{\Gamma(L)} =1$.
Then the limit CDF of $\text{SINR}_k$ in \eqref{CDF_approx} is of the form
\begin{align} \label{SINR_limit_X_prime}
 \lim_{P_t \to \infty} F_{\text{SINR}_k} (\gamma) 
    & = 1 - F_{X'}( \xi ), 
\end{align}
where the exchange of the limit and the integral is allowable because the integral is finite and the integrand is non-negative, and where $F_{X'}(x)$ denotes the CDF of $X'$.
So we can derive the limit CDF of $((K-1) \text{SINR}_k)^{-1}$ as
\begin{align}
    &\lim_{P_t \to \infty} \Pr\left\{  \frac{1}{(K-1) \text{SINR}_k} \le x  \right\} \notag\\ 
    &\hspace{0.5cm}= 1 -  \lim_{P_t \to \infty}  F_{\text{SINR}_k}\left(  \frac{1}{x (K-1)}  \right) 
    \overset{(a)}{=} F_{X'}(x),
\end{align}
where $(a)$ follows from \eqref{SINR_limit_X_prime} after considering $X' \in [0,1]$.
The above directly leads to the convergence result in \eqref{SINRk_dist_limit}.
\end{proof}

\subsection{Massive MIMO Regime}
In what follows, we consider the massive MIMO regime where $L$ goes to infinity while $K$ keeps finite\footnote{We note that if $L$ and $K$ go to infinity with a fixed ratio, $\text{SINR}_k$ will converge to a constant almost surely (cf. Lemma~\ref{EC_massive_lem}).}.
\begin{lemma}\label{SINRk_lim_lem}
    In massive MIMO, we have that:
    \begin{align}
        \frac{\text{SINR}_k}{L} \stackrel{d.}{\longrightarrow}  \Big(\frac{K \sigma_k^2}{P_t} + \text{Gamma}(K-1,1) \Big)^{-1}.
    \end{align}
\end{lemma}

\begin{proof}
We first rewrite the expression for $\text{SINR}_k$ in \eqref{SINR_divide_eq} as
\begin{align}
    \text{SINR}_k = \frac{ L  }{ \frac{K \sigma_k^2}{P_t||{\bf h}_k||^2/L} +  \sum_{i=1, i \neq k}^K  L X_i}
\end{align}
The Strong Law of Large Numbers (SLLN) \cite{Papoulis} tells us that
\begin{align}\label{Y_lim_L_eq}
    \frac{K \sigma_k^2}{P_t ||{\bf h}_k||^2/L}  \stackrel{a.s.}{\longrightarrow}   \frac{K \sigma_k^2}{P_t}, \text{  as } L \to \infty.
\end{align}
For the CDF of $L X_i$, we have that
\begin{align}
   \lim_{L \to \infty} F_{L X_i}(x) 
   &\overset{(a)}{=} 1-  \lim_{L \to \infty}   \Big(1-\frac{x}{L} \Big)^{L-1} \mathbb{I}\Big\{ \frac{x}{L} \le 1 \Big\} \notag \\
   &= 1- \exp(-x),
\end{align}
where $(a)$ follows from the CDF in Lemma \ref{Xi_Distri_Lem}, and where $\mathbb{I}\{\cdot\}$ denotes the well-known indicator function.
The above indicates that $L X_i$ converges to an exponentially distributed random variable $X_i'$ with unit-mean in distribution.\footnote{
We note that the fact of $X_n \stackrel{d.}{\longrightarrow} X_\infty $ and $Y_n \stackrel{d.}{\longrightarrow} Y_\infty$ does not generally imply that $X_n + Y_n \stackrel{d.}{\longrightarrow} X_\infty + Y_\infty$.} However,  $L X_i$ does not converge to $X_i' \sim \text{Exp}(1)$ in probability. For any $\epsilon>0$, we have that
\begin{align}\label{Xi_prop_conver}
    & \Pr\left\{ \big| X_i' - L X_i \big| \le \epsilon \right\}
     = \Pr\left\{L X_i -  \epsilon  \le  X_i' \le L X_i +\epsilon \right\} \notag\\
    &\hspace{0.5cm}= \mathbb{E}_{X_i}\Big\{  F_{X_i'} (L X_i +\epsilon) - F_{X_i'}(L X_i -  \epsilon)  \Big\}\notag\\
    &\hspace{0.5cm} \overset{(a)}{=}  \mathbb{E}_{X_i}\Big\{ \exp\big( - (L X_i -\epsilon) \big)  - \exp\big( - (L X_i +\epsilon) \big)  \Big\} \notag\\
    &\hspace{0.5cm} = \big( \exp(\epsilon) - \exp(-\epsilon) \big) \mathbb{E}_{X_i}\Big\{ \exp\big( - L X_i  \big)  \Big\} \notag\\
    &\hspace{0.5cm}\overset{(b)}{\le} \exp(\epsilon) - \exp(-\epsilon),
\end{align}
where $(a)$ follows from the exponential distribution of $X_i'$, and $(b)$ follows from the fact that $\exp(-x)\le 1$ for any $x\ge 0$. The above shows that $LX_i$ does not converge to an exponential distribution in probability.  For example, when $\epsilon=0.1$, the probability in  \eqref{Xi_prop_conver} is always upper bounded by 0.2003..., which is strictly less than 1 regardless of how large $L$ is.

The fact that $L X_i \stackrel{d.}{\longrightarrow} \text{Exp}(1)$ implies the limit CF of $L X_i$ equaling the CF of $\text{Exp}(1)$, i.e.,
$
    \lim_{L \to \infty} \text{CF}_{L X_i}(t) 
    = ( 1 - \jmath t )^{-1}.
$
Then due to the independence among $\{X_i\}_{i=1,i\neq k}^K$ (cf. Lemma~\ref{Xi_Distri_Lem}), we can derive the limit CF of $L X \triangleq \sum_{i=1, i \neq k}^K L X_i$ as
\begin{align}
    \lim_{L \to \infty} \text{CF}_{L X}(t) &=  \prod_{i=1, i\neq k}^K \lim_{L \to \infty} \text{CF}_{L X_i} (t) 
    = \big( 1 - \jmath t \big)^{-(K-1)}, \notag
\end{align}
which exactly equals the CF of $\text{Gamma}(K-1,1)$. So we have the convergence result:
\begin{align}\label{Xi_sum_lim_eq}
    \sum_{i=1, i\neq k}^K L X_i \stackrel{d.}{\longrightarrow} \text{Gamma}(K-1, 1), \text{ as } L \to \infty.
\end{align}
Combining \eqref{Y_lim_L_eq} and \eqref{Xi_sum_lim_eq} and using the Continuous Mapping Theorem \cite{Vaart_book} finally yields Lemma \ref{SINRk_lim_lem}.
\end{proof}

\subsection{Beta Approximation}
As the exact distribution of the sum of multiple i.i.d. Beta distributed random variables has not yet been revealed, we alternatively consider a new Beta distributed random variable to approximate this sum \cite{Arjun}.  Specifically, $X' \triangleq \frac{1}{K-1} \sum_{i=1, i \neq k}^K X_i \in [0,1]$ is approximately modelled by a Beta distribution with the first shape parameter $\alpha$ and the second shape parameter $\beta$. Then, we use ${\rm B}(\cdot,\cdot)$ and $\Gamma(\cdot,\cdot)$ to respectively denote the Beta function and the upper incomplete Gamma function \cite{Gradshteyn}. A general approximation is shown in Lemma~\ref{SINRk_beta_lem}, where $\xi \triangleq \min\big\{ \frac{1}{(K-1)\gamma}, 1 \big\}$ for any $\gamma>0$.

\begin{lemma}\label{SINRk_beta_lem}
The CDF of $\text{SINR}_k$ can be approximated as
    \begin{align}\label{SINRk_approx_eq}
     &F_{\text{SINR}_k} (\gamma) 
     \approx  1 - \frac{1}{\Gamma(L) {\rm B}(\alpha, \beta)} \int_0^{\xi} x^{\alpha-1} (1 - x)^{\beta -1}  \notag\\
         &\hspace{2.2cm} \times \Gamma\bigg(L, \frac{K \sigma_k^2}{P_t} \Big( \frac{1}{\gamma} - (K-1) x \Big)^{-1} \bigg)  \dd x,
    \end{align}
where $\alpha$ and $\beta$ respectively take the forms
\begin{align}\label{alpha_beta_eq}
    \alpha = \frac{(K-1)(L+1) - 1}{L}, \ \  \ 
    \beta = \alpha (L-1).
\end{align}
\end{lemma}

\begin{proof}
To approximately model $X'$ as a Beta distribution, we need to match the first two moments of $X'$. For that, we can derive $\alpha$ and $\beta$ by the following:
\begin{align}
    &\mathbb{E}\{X'\}  = \mathbb{E}\Big\{\frac{1}{K-1} \sum_{i \neq k} X_i \Big\}   = \frac{1}{L} \notag\\
    & \text{Var}\{X'\} = \text{Var}\Big\{ \frac{1}{K-1} \sum_{i \neq k} X_i \Big\}   = \frac{(L-1)}{L^2 (K-1) (L+1)}, \notag
\end{align}
which easily leads to the expressions for $\alpha$ and $\beta$ in \eqref{alpha_beta_eq}.

By approximately modelling $X'$ as $\text{Beta}(\alpha,\beta)$ and substituting the corresponding PDF into \eqref{CDF_approx}, we can finally derive \eqref{SINRk_approx_eq}, which concludes the proof.
\end{proof}

\begin{remark}
    The integral in \eqref{SINRk_approx_eq} is finite and the integrand is continuous and real-valued over $[0,\xi]$. Furthermore, as $L$ is a positive integer, we can express $\Gamma(L,x)/\Gamma(L)$ as $ \exp(-x) \sum_{\vartheta=0}^{L-1}\frac{x^\vartheta}{\vartheta !}$ (elementary functions). Therefore, we can find the numerical solution of \eqref{SINRk_approx_eq} very efficiently. 
\end{remark}

\begin{remark}
When we approximately model $X'$ as a Beta distribution in Lemma~\ref{SINRk_dist_limit_lem}, a direct approximation for $F_{\text{SINR}_k}(\gamma)$ in high SNR is derived as 
\begin{align}\label{SINRk_beta_limit_eq}
     F_{\text{SINR}_k} (\gamma) & \approx 1 - \text{I}_\xi (\alpha, \beta),
\end{align}
where $\text{I}_\cdot(\cdot,\cdot)$ denotes the regularized incomplete Beta function \cite{Gradshteyn}, which is used to express the CDF of Beta distribution. 
\end{remark}

\section{Ergodic Rate Analysis}
In this section, we analyze the ergodic rate under MF precoding. Due to the space limitation, we only present two main results: $i)$ a convergence result in massive MIMO, and $ii)$  a robust approximation.

Let us consider the massive MIMO regime where $L, K \to \infty$ and the ratio $c = \frac{L}{K}$ is fixed. 
\begin{lemma}\label{EC_massive_lem}
As $L,K \to \infty$ with a fixed ratio $c$, we have the convergence result:
    \begin{align}\label{rate_lim_eq}
    R_k \stackrel{a.s.}{\longrightarrow}  \ln\Big( 1 + \frac{c \ P_t}{P_t + \sigma_k^2} \Big).
\end{align}
\end{lemma}

\begin{proof}
First of all, we rewrite $\text{SINR}_k$ in \eqref{SINR_divide_eq} as
\begin{align}
    \text{SINR}_k =  \frac{P_t L}{K}   \Big( \frac{\sigma_k^2}{||{\bf h}_k||^2/L} + \frac{P_t}{K} \sum\nolimits_{{i=1, i \neq k}}^K  L X_i  \Big)^{-1}.
\end{align}
Thanks to Lemma~\ref{Xi_Distri_Lem}, we can easily have that $\{L X_i\}_{i=1,i\neq k}^K$ are i.i.d. random variables, each with unit-mean and finite variance. Then, as $L,K\to \infty$, we can use the SLLN to derive that
\begin{align}
    \frac{1}{K-1} \sum\nolimits_{i=1,i\neq k}^K L X_i \stackrel{a.s.}{\longrightarrow} 1, \text{ and } 
    \frac{||{\bf h}_k||^2}{L}  \stackrel{a.s.}{\longrightarrow} 1.
\end{align}
Using the Continuous Mapping Theorem finally yields \eqref{rate_lim_eq}.
\end{proof}

\begin{remark}
    The convergence result in \eqref{rate_lim_eq} equals the well-known asymptotic rate under MF precoding in massive MIMO (cf. \cite{Rusek}). It is indeed the first time to \emph{rigorously} prove how the exact rate converges to the asymptotic rate.\footnote{This asymptotic rate was previously derived from various approximation methods, e.g., treating either the useful signal or the interference to be deterministic  (cf. \cite{Yeon_Geun_Lim, Feng}), or separately taking their expectations regardless of the correlations (cf. \cite{Zhang}). Interestingly, using Jensen's Inequality (cf. \cite{Hien,Zhao_SPAWC}) can also lead to this asymptotic rate.} 
\end{remark}

Next, we propose a robust approximation in general cases.
Let $\mu_Z$ and $\sigma_Z^2$ denote the mean and the variance of $Z \triangleq Y + \sum_{i=1,i\neq k}^K X_i$ respectively. With the help of Lemma~\ref{Xi_Distri_Lem}, we can easily  derive that
\begin{align}
    &\mu_Z =  \frac{K \sigma_k^2}{P_t(L-1)} + \frac{K-1}{L} \label{mu_Z_eq}\\
    &\sigma_Z^2 = \frac{K^2 \sigma_k^4}{P_t^2 (L-1)^2 (L-2)} +  \frac{(K-1)(L-1)}{L^2 (L+1)} \label{sigmaZ_eq}
\end{align}
Now, we present a robust approximation for the ergodic rate.
\begin{lemma}\label{Robust_lem}
For $L>2$, we can robustly approximate $\bar R_k $ as
\begin{align}\label{R_k_robust_eq}
    \bar R_k \approx \ln\Big( 1 + \frac{1}{\mu_Z} \Big) + \frac{\sigma_Z^2}{2} \frac{2 \mu_Z + 1}{\mu_Z^2 (\mu_Z+1)^2}.
\end{align}
\end{lemma}

\begin{proof}
The proof follows from the approximation method in \cite{Holtzman,Zhao2020_TWC,Zhao_GK_CL} by considering that $\bar R_k = \mathbb{E}\{ \ln ( 1 + \frac{1}{Z} ) \}$.
\end{proof}

\section{Numerical Results}
In this section, we will present some numerical results to demonstrate the accuracy of the developed analytical models. 

Fig.~\ref{OP_fig_ref} plots the outage probability at the $k$-th user for different numbers of transmit antennas. As expected, the outage probability first decreases and then arrives at a floor as $P_t$ increases because MF precoding is interference-limited. The increase in the number of transmit antennas improves the outage performance due to more spatial diversity brought about by a larger-scale antenna array. Importantly, both the exact and approximate results derived separately from Theorem~\ref{SINR_MF_dist_thm} and Lemma~\ref{SINRk_beta_lem} always match the simulated results very well.

Fig.~\ref{CDF_lim_fig_ref} (left) plots the CDF of $( (K-1) \text{SINR}_k )^{-1}$ for different transmit powers, where the red dashed line represents the CDF of $X'$ in Lemma~\ref{SINRk_dist_limit_lem}. It is obvious that as $P_t$ increases, $( (K-1) \text{SINR}_k )^{-1}$ finally converges to the distribution of $X'$, which exactly obeys the convergence statement in Lemma~\ref{SINRk_dist_limit_lem}. We plot the CDF of $\text{SINR}_k/L$ by varying the number of transmit antennas in Fig.~\ref{CDF_lim_fig_ref} (right), where the red dashed line represents the limit distribution in Lemma~\ref{SINRk_lim_lem}. The convergence of the blue symbol lines to the red dashed line becomes apparent as $L$ increases, which demonstrates the accuracy of Lemma~\ref{SINRk_lim_lem}. 

\begin{figure}[!t]
             \vspace{-0.2cm}
             \centering
             \includegraphics[width= 3.5 in]{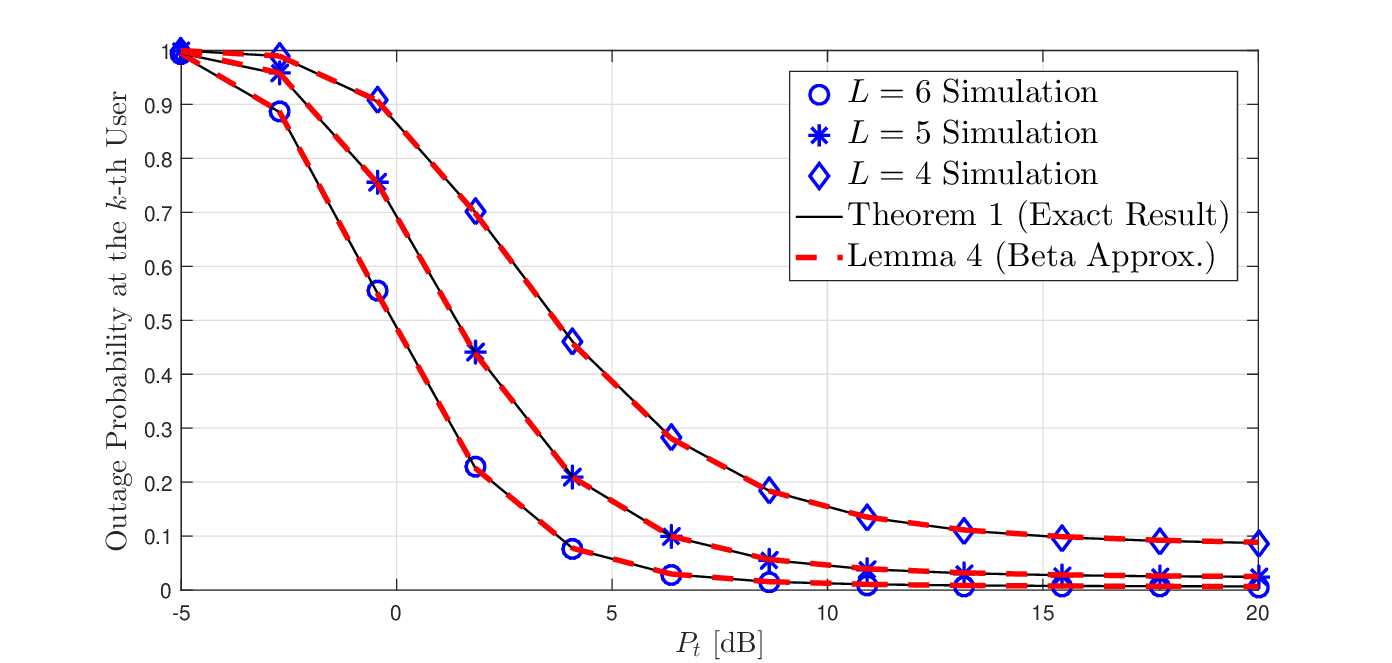}
             \vspace{-0.5cm}
             \captionsetup{font={footnotesize}}
        \caption{Outage probability versus $P_t$ for $\gamma =0.8$, $K=4$ and $\sigma_k^2=1$.}\label{OP_fig_ref}
\end{figure}	

             \begin{figure}[t]\centering\vspace{-0.3cm}
				\begin{subfigure}[t]{.12\textwidth}\centering		
				\includegraphics[width=1.6 in]{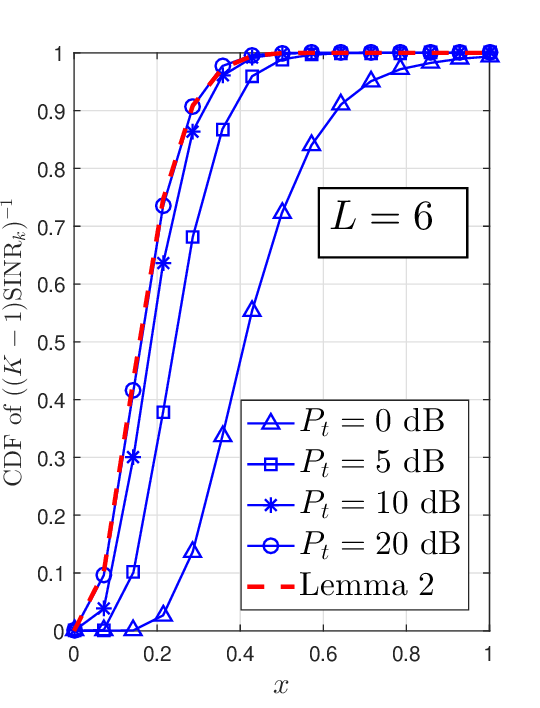}
				\end{subfigure}~
				\begin{subfigure}[t]{.45\textwidth}
                  \centering				
				\includegraphics[width=1.6 in]{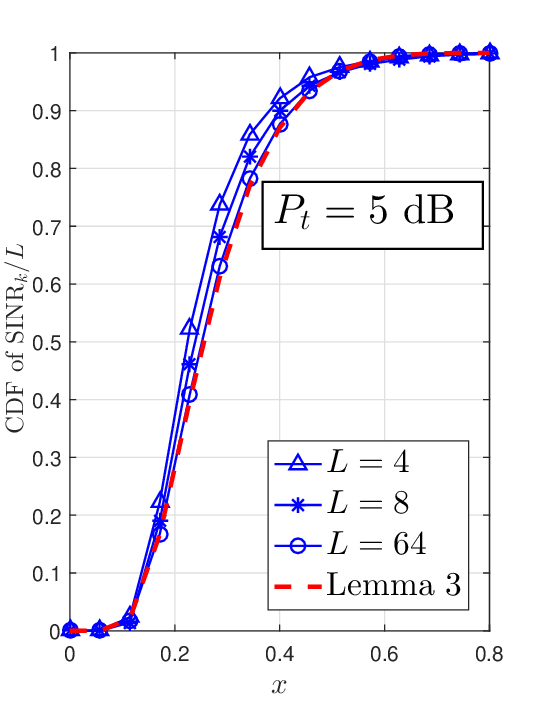}
				\end{subfigure}~\vspace{-1.5ex}
                \captionsetup{font={footnotesize}}
			\caption{CDFs of $\frac{1}{(K-1) \text{SINR}_k}$ and $\frac{\text{SINR}_k}{L}$ for $K=4$ and $\sigma_k^2=1$.}\label{CDF_lim_fig_ref}\vspace{-0.2cm}
		\end{figure}

The ergodic rate for the $k$-th user is plotted in Fig.~\ref{EC_fig_ref}. We note that the green triangle line and the red dashed line respectively stand for the results separately derived from Lemmas~\ref{EC_massive_lem} and \ref{Robust_lem}, while the result in the black star line represents a lower bound of the ergodic rate derived by using Jensen's Inequality, given by $\ln(1+1/\mathbb{E}\{\text{SINR}_k^{-1}\})$. We can observe that all the kinds of approximations perform well in low SNR, while both the Jensen's bound and the asymptotic rate fail to match the simulated result in the medium to high SNR regime, even when we increase the number of antennas from 8 to 12. In contrast, the robust approximation in Lemma~\ref{Robust_lem} always presents a high accuracy to the simulation in the whole SNR regime.  

\begin{figure}[!t]
             \centering
             \includegraphics[width= 3.5 in]{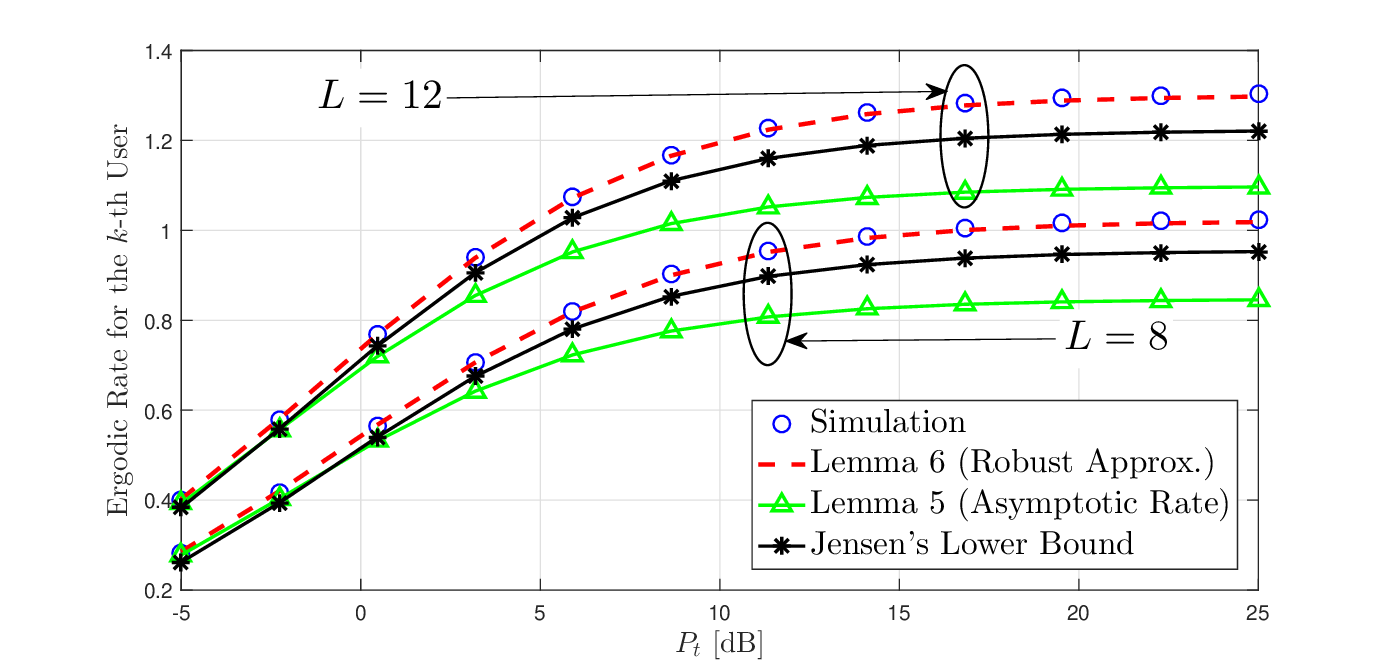}
             \vspace{-0.5cm}
             \captionsetup{font={footnotesize}}
        \caption{Ergodic rate versus $P_t$ for $K=6$ and $\sigma_k^2=1$.}\label{EC_fig_ref}\vspace{-0.3cm}
\end{figure}	

\section{Conclusion}
We revealed the exact SINR distribution under MF precoding for the first time, as well as developed several good approximations. Specifically, we rigorously proved that the SINR would converge to some specific distributions in high SNR and massive MIMO. Then we shifted to the ergodic rate analysis by deriving a robust approximation and a convergence result in massive MIMO where the latter told us that the exact rate converges to the well-known asymptotic rate almost surely. Finally, numerical results have demonstrated the accuracy of the analytical models. 

\appendices
\renewcommand{\thesectiondis}[2]{\Alph{section}:}

\section{Proof of Lemma~\ref{Xi_Distri_Lem}}\label{Proof_Xi_Distri_Lem}
The CDF of $X_i \in [0,1]$ can be written as
\begin{align}
    F_{X_i} (x) &= \Pr\left\{ \frac{ 1}{|| {\bf h}_i^* ||^2}   \big| {\bf u}_k^T {\bf h}_i^*  \big|^2 \le x \right\} \notag\\
    & = \Pr\left\{ {\bf h}_i^T {\bf u}_k^* {\bf u}_k^T {\bf h}_i^*  \le x  {\bf h}_i^T {\bf h}_i^* \right\} \notag\\
    &= \Pr\left\{    {\bf h}_i^T  \Big(  {\bf I}_L - \frac{1}{x} {\bf u}_k^* {\bf u}_k^T \Big) {\bf h}_i^* \ge 0 \right\}. \label{CDF_Xi_def}
\end{align}
We note that $\text{Rank}\big\{ \frac{1}{x} {\bf u}_k^* {\bf u}_k^T \big\} =1$.
We then eigen-decompose ${\bf I}_L - \frac{1}{x} {\bf u}_k^* {\bf u}_k^T $ (Hermitian matrix) as
\begin{align}
    {\bf I}_L - \frac{1}{x} {\bf u}_k^* {\bf u}_k^T  = {\bf V}_k {\bf \Omega} {\bf V}_k^H,
\end{align}
where the columns of ${\bf V}_k$ (each with unit-norm) are the eigenvectors of ${\bf I}_L - \frac{1}{x} {\bf u}_k^* {\bf u}_k^T$, and the diagonal matrix ${\bf \Omega} \in \mathbb{C}^{L \times L}$ is of the form
\begin{align}
    {\bf \Omega} = \text{Diag} \Big\{1-\frac{1}{x}, 1,1,\cdots,1  \Big\},
\end{align}
whose diagonal elements are the eigenvalues of ${\bf I}_L - \frac{1}{x} {\bf u}_k^* {\bf u}_k^T$.
Obviously, $ {\bf \Omega}$ is independent of ${\bf h}_k$. As ${\bf V}_k$ is a unitary matrix, the columns of ${\bf V}_k$ can span the linear space of $\mathbb{C}^{L \times 1}$.
Therefore, for any ${\bf h}_i^* \in \mathbb{C}^{L \times 1}$, we can always find a \emph{unique} vector ${\bf a}_i = [a_1, a_2, \cdots,a_L]^T \in \mathbb{C}^{L \times 1}$ such that
$
    {\bf h}_i^* = {\bf V}_k {\bf a}_i,
$
where ${\bf V}_k$ and ${\bf h}_i^*$ are independent.
As ${\bf V}_k = [{\bf v}_1, {\bf v}_2, \cdots, {\bf v}_L]$ is a unitary matrix, we can derive ${\bf a}_i$ as
\begin{align}
    {\bf a}_i = {\bf V}_k^H {\bf h}_i^*. 
\end{align}
Given that ${\bf h}_i^* \sim \mathcal{CN}({\bf 0}, {\bf I}_L)$,  each element $a_\ell = {\bf v}_\ell^H {\bf h}_i^*$ of ${\bf a}_i$ will be also Gaussian distributed with zero-mean and unit-variance. Furthermore, as ${\bf v}_1^H, {\bf v}_2^H, \cdots, {\bf v}_L^H$ are \emph{orthonormal} to each other, the elements of ${\bf a}_i$ are independent. This can be easily proved. For any $\ell' \neq \ell$, we have that
\begin{align}
  \mathbb{E}\{ a_{\ell'} a_\ell^* \} 
  =  {\bf v}_{\ell'}^H \mathbb{E}\{ {\bf h}_i^*  {\bf h}_i^T \} {\bf v}_\ell = {\bf v}_{\ell'}^H {\bf v}_\ell  = 0.
\end{align}
So $a_\ell$ and $a_{\ell'}$ are uncorrelated. As $a_\ell$ and $a_{\ell'}$ are both Gaussian distributed, they are independent. We also note that ${\bf a}_i$ is independent of ${\bf V}_k$ (or equivalently, ${\bf h}_k$) because ${\bf V}_k$ does not affect the distribution of ${\bf a}_i$ as long as ${\bf V}_k$ is always a unitary matrix.

By doing so, we can rewrite ${\bf h}_i^T  \big(  {\bf I}_L - \frac{1}{x} {\bf u}_k^* {\bf u}_k^T \big) {\bf h}_i^*$ as
\begin{align}
    &{\bf h}_i^T  \Big(  {\bf I}_L - \frac{1}{x} {\bf u}_k^* {\bf u}_k^T \Big) {\bf h}_i^*  
    = {\bf a}_i^H {\bf V}_k^H  \Big(  {\bf I}_L - \frac{1}{x} {\bf u}_k^* {\bf u}_k^T \Big) {\bf V}_k {\bf a}_i \notag\\
    &\hspace{1.5cm}  = {\bf a}_i^H {\bf \Omega} {\bf a}_i = |a_1|^2 \Big(1-\frac{1}{x}\Big) + \sum_{\ell=2}^L |a_\ell|^2.
\end{align}
We note that $|a_\ell|^2$ follows an exponential distribution with unit-mean given that $a_\ell \sim \mathcal{CN}(0,1)$.
Therefore, we can rewrite the CDF of $X_i$ in \eqref{CDF_Xi_def} as
\begin{align}
    F_{X_i} (x) 
    = \Pr\left\{ \frac{|a_1|^2}{|a_1|^2 + \sum_{\ell=2}^L |a_\ell|^2} \le x  \right\}.
\end{align}
As the probability of $\Pr\{X_i\le x\}$ is determined \emph{solely} by ${\bf a}_i$ which is independent of ${\bf h}_k$, we can conclude that $X_i$ is also independent of ${\bf h}_k$.
It is easy to derive that $\sum_{\ell=2}^L |a_\ell|^2 \sim \text{Gamma}(L-1,1)$, and then  $\frac{|a_1|^2}{|a_1|^2 + \sum_{\ell=2}^L |a_\ell|^2}$ has a beta distribution with the first shape parameter $1$ and the second shape parameter $L-1$, which leads to the CDF and PDF of $X_i$ in Lemma~\ref{Xi_Distri_Lem}.  

For any $j \in \{1,2,\cdots,K\}$ and $j \neq k,i$, the Gaussian distributed vectors ${\bf a}_i = {\bf V}_k^H {\bf h}_i^*$ and ${\bf a}_j = {\bf V}_k^H {\bf h}_j^*$ are independent because ${\bf h}_i^*$ and ${\bf h}_j^*$ are independently Gaussian distributed where $\mathbb{E}\{ {\bf h}_i^*   {\bf h}_j^T \} = {\bf 0}$. This can be observed by the fact that
$
    \mathbb{E}\{  {\bf a}_i {\bf a}_j^H\} 
    = \mathbb{E}\{ {\bf V}_k^H \mathbb{E}\{ {\bf h}_i^*   {\bf h}_j^T \}  {\bf V}_k \} = {\bf 0}.
$
This also indicates that $X_i$ and $X_j$ are independent.

\section{Proof of Theorem~\ref{SINR_MF_dist_thm}}\label{Proof_SINR_MF_dist_thm}
As $X_i \sim \text{Beta}(1,L-1)$, the CF of $X_i$ is of the form \cite{Papoulis}
\begin{align}
    \text{CF}_{X_i}(t) &\triangleq \mathbb{E}\{ \exp(\jmath t X_i)  \} \notag\\
           & = (L-1) \exp(\jmath t) (\jmath t)^{1-L} \Upsilon(L-1, \jmath t).
\end{align}
Let $ X \triangleq \sum_{i=1, i \neq k}^K  X_i$, where the summation terms are independent of each other (cf. Lemma~\ref{Xi_Distri_Lem}). Therefore, we can derive the CF of $X$ as
\begin{align}
    &\text{CF}_X(t) = \prod\nolimits_{i=1, i \neq k}^K \text{CF}_{X_i}(t) \notag\\
    &= (L-1)^{K-1} \frac{\exp(\jmath (K-1) t )}{ (\jmath t)^{(L-1)(K-1)} } \Big[ \Upsilon(L-1, \jmath t) \Big]^{K-1}.
\end{align}

Let $Z \triangleq Y + X $. 
We rewrite the CDF of $\text{SINR}_k$ in \eqref{SINR_divide_eq} as
\begin{align}\label{CDF_general}
    F_{\text{SINR}_k}(\gamma) &=  \Pr\left\{ \frac{ 1  }{ Y +  X} \le \gamma \right\}
    = 1 - F_Z\Big(\frac{1}{\gamma} \Big)
\end{align}
 As  $Y \sim \text{Inv-Gamma}(L,\frac{K\sigma_k^2}{P_t})$, its CF is given by (cf. \cite{Papoulis})
\begin{align}
    \text{CF}_{Y}(t) = \frac{2(-\jmath \frac{K\sigma_k^2 t}{P_t})^{L/2}}{\Gamma(L)} \mathcal{K}_L\bigg( \sqrt{\frac{-4\jmath  K\sigma_k^2 t}{P_t}} \bigg).
\end{align}
As Lemma \ref{Xi_Distri_Lem} has proved that $X$ and $Y$ are independent, we can express the CF of $Z$ as
\begin{align}
    &\text{CF}_Z(t) = \text{CF}_{Y}(t) \text{ CF}_X(t) \notag\\
   &\hspace{1cm} =  \frac{2(L-1)^{K-1} }{\Gamma(L)} \Big(\frac{K\sigma_k^2 }{P_t} \Big)^{\!\! \frac{L}{2}} \frac{ \jmath^L\exp(\jmath (K-1) t )} {(\jmath t)^{(L-1)(K-1)-L/2} }  \notag\\
   &\hspace{1.5cm}\times \mathcal{K}_L\bigg( \sqrt{\frac{-4\jmath  K\sigma_k^2 t}{P_t}} \bigg)   \Big[ \Upsilon(L-1, \jmath t) \Big]^{K-1}.
\end{align}
Based on the CF of $Z$, the PDF of $Z$ can be derived by \cite{Papoulis}
\begin{align}
    f_Z(z) &=  \frac{1}{2 \pi} \int_{-\infty}^{+\infty} \exp(-\jmath t z)  \text{ CF}_{Z}(t) \dd t. 
\end{align}
According to Gil-Pelaez Theorem \cite{Gil}, we can derive the CDF of $Z$  by
\begin{align}
    F_Z(z) =  \frac{1}{2} - \frac{1}{\pi} \int_0^\infty \text{Im}\left\{ \frac{\exp(-\jmath t z)}{t}  \text{ CF}_Z(t) \right\} \dd t.  
\end{align}

Combining \eqref{CDF_general} and the CDF of $Z$, we can derive the CDF of $\text{SINR}_k$ in 
Theorem \ref{SINR_MF_dist_thm}. Then
the PDF of $\text{SINR}_k$ can be derived by
\begin{align}
    f_{\text{SINR}_k} (\gamma) =\frac{\partial F_{\text{SINR}_k}(\gamma)}{\partial \gamma}= \frac{1}{\gamma^2} f_Z\Big(\frac{1}{\gamma} \Big) 
\end{align}
which yields \eqref{PDF_SINRk_eq} by using the PDF of $Z$.

	\bibliographystyle{IEEEtran}				
	\bibliography{IEEEabrv,aBiblio}			




\end{document}